\documentclass{article}

\usepackage{amsmath,amsfonts,bm}

\newcommand{\F}{\mathcal{F}}
\newcommand{\Fr}{\mathcal{F}_r}

\def\eqref#1{Equation (\ref{#1})}

\def\1{\bm{1}}

\DeclareMathAlphabet{\mathsfit}{\encodingdefault}{\sfdefault}{m}{sl}
\SetMathAlphabet{\mathsfit}{bold}{\encodingdefault}{\sfdefault}{bx}{n}

\newcommand{\E}{\mathbb{E}}

\newcommand{\R}{\mathbb{R}}

\newcommand{\Var}{\mathrm{Var}}

\DeclareMathOperator*{\argmin}{arg\,min}

\usepackage[natbib, style=authoryear-comp]{biblatex}
\usepackage{hyperref}
\usepackage{url}
\usepackage{amsmath, amssymb, amsthm}
\usepackage{breqn}
\usepackage{graphicx}
\usepackage{caption}
\usepackage{subcaption}
\usepackage{amssymb, amsmath, amsthm, mathtools}

\theoremstyle{definition}
\newtheorem{theorem}{Theorem}[section]

\newtheorem{assumption}[theorem]{Assumption}
\newtheorem{proposition}[theorem]{Proposition}

\newtheorem{lemma}[theorem]{Lemma}
\newtheorem{remark}[theorem]{Remark}

\makeatletter
\def\munderbar#1{\underline{\sbox\tw@{$#1$}\dp\tw@\z@\box\tw@}}
\makeatother

\newcommand{\be}{\begin{equation}}
\newcommand{\ee}{\end{equation}}

\newcommand{\mc}{\mathcal}

\newcommand{\inner}[2]{\big \langle #1, #2 \big \rangle }

\DeclarePairedDelimiter{\smallnorm}{\lVert}{\rVert}
\newcommand{\smallinner}[2]{\langle #1, \ #2 \rangle}

\newcommand{\one}{\textbf{1}}

\newcommand{\td}{\widetilde}

\newcommand{\norm}[1]{\left\|#1\right\|}

\newcommand{\estimand}{\psi}

\newcommand{\ignore}[1]{}

\usepackage{xcolor}

\title{A Stable and Efficient Covariate-Balancing Estimator for Causal Survival Effects}

\makeatletter
\newcommand{\myfnsymbol}[1]{%
  \expandafter\@myfnsymbol\csname c@#1\endcsname
}
\newcommand{\@myfnsymbol}[1]{%
  \ifcase #1
  \or 1
  \or 2
  \or 3
  \or 4
  \fi
}
\newcommand{\affiliationA}{\@myfnsymbol{1}}
\newcommand{\affiliationB}{\@myfnsymbol{2}}
\newcommand{\affiliationC}{\@myfnsymbol{3}}
\newcommand{\affiliationD}{\@myfnsymbol{4}}

\author{
  Khiem Pham \textsuperscript{\affiliationA},
  David A. Hirshberg\textsuperscript{\affiliationB},
  Phuong-Mai Huynh-Pham\textsuperscript{\affiliationA},\\
  Michele Santacatterina\textsuperscript{\affiliationC},
  Ser-Nam Lim\textsuperscript{\affiliationD},
  Ramin Zabih\textsuperscript{\affiliationA}
}

\begin{document}
\maketitle
\footnotetext[1]{Cornell University}%
\footnotetext[2]{Emory University}%
\footnotetext[3]{New York University}%
\footnotetext[4]{University of Central Florida}%

\begin{abstract}

We propose an empirically stable and asymptotically efficient covariate-balancing approach to the problem of estimating survival causal effects in data with conditionally-independent censoring. This addresses a challenge often encountered in state-of-the-art nonparametric methods: the use of inverses of small estimated probabilities and the resulting amplification of estimation error. We validate our theoretical results in experiments on synthetic and semi-synthetic data.


\end{abstract}

\section{Introduction and Related Work}

Estimating the impact of interventions on survival times is a key objective in numerous studies, spanning domains such as drug efficacy's evaluation in terms of ICU stay duration and the assessment of advertising campaigns' effects on customer dwell time. To measure this causal effect, the key survival object is generally the \textit{counterfactual survival curve} \citep{westling2023inference}, which represents the probability an
an individual in our population would experience an event after a specific point in time if, possibly contrary to fact, they had undergone a specific intervention. 
In this paper, we focus specifically on discrete time to events \citep{stitelman2011targeted, curth2021survite}.

Traditional approaches for estimating counterfactual survival curves have relied on parametric or semiparametric regression models, including marginal structural models \citep{kleinbaum2012parametric,cox1972regression, tsiatis2006semiparametric,yiu2022joint}. These models, however, struggle to capture complex relationships within the data. Recently, it has become popular to use machine learning methods to capture these complex covariate-outcome relationships \citep[e.g][]{ishwaran2008random, zhu2016deep,katzman2018deepsurv,ching2018cox,hu2021transformer, rindt2022survival}.

These methods can suffer from bias due to regularization. However, their use in combination with bias corrections arising from semiparametric/nonparametric efficiency theory can substantially ameliorate this problem \citep{bickel1993efficient,chernozhukov2018double, zheng2010asymptotic}. These corrections, in this context, are based on inverse probability weighting. In particular, in survival settings, it is necessary to weight observations by the inverse of the product of the conditional probability of treatment assignment and censoring \citep{robins1992recovery}. This poses a challenge for research that follows this approach, e.g.  \citet{westling2023inference} and \citet{cui2023estimating}, because small errors in the estimation of small probabilities can lead to large errors in the estimation of their inverses.

In contrast, we generalize techniques from the covariate balancing literature to perform this weighting adjustment \citep{imai2014covariate, hainmueller2012entropy, zubizarreta2015stable,li2018balancing,ben2021balancing, wong2017kernel,kallus2022optimal,hirshberg2021augmented}. Our proposed estimator is stable in small sample sizes, and is asymptotically efficient. That is, it is approximately normal with negligible bias, justifying the use of standard $\hat\mu \pm 2 \hat\sigma$ confidence intervals, and moreover has optimal variance among all such estimators. While analogous results have been established for cross-sectional data, the generalization to survival data is nontrivial, as adjusting for censoring in this setting involves working with time-varying conditioning structure.
Existing work in this direction is limited. \citet{xue2023rkhs} proposed an approach that adjusts only for selection into treatment, assuming independent censoring and \citet{kallus2021optimal} proposed an approach based on marginal structural models. Furthermore, \citet{kallus2021optimal} provided no theoretical guarantees and \citet{xue2023rkhs} established only a rate of convergence. 

In the literature on automatic debiasing \citep[e.g.][]{chernozhukov2022automatic}, 
approaches that directly estimate debiasing weights for longitudinal problems, 
like the aforementioned product of inverse assignment and censoring probabilities,
have been proposed \citep[e.g.][]{chernozhukov2022automatic}. Our framing and analysis of the
problem differs in that we focus on balance as an in-sample property rather than a population
property arising from the accurate estimation of a function outputting debiasing weights.
On a theoretical level, the generalization of this conceptual and analytic approach 
from cross-sectional to longitudinal settings follows a different path from the 
analogous generalization of the automatic debiasing approach. And on a practical level, 
the differences are largely in the details, 
e.g. in that we optimize our weights for in-sample balance rather than cross-fit them by 
optimizing a similar criterion in a different fold and tune so that, if we think of our 
weights as estimating this weight-outputting function, we are undersmoothing. These differences between
balancing and automatic debiasing approaches persist from context to context,
e.g. in the discussion of estimation of treatment effects in cross-sectional data
using RKHS models \citep[e.g.][]{hirshberg2019minimax, singh2021debiased}. To our
knowledge, the effect of these differences on performance has not been conclusively understood theoretically or empirically.\footnote{See \citet{bruns2023augmented}, however, for some discussion of undersmoothing.}

Our paper is organized as follows: In section \ref{sec:prelim}, we introduce the notation and assumptions used throughout our work. In particular, we use standard assumptions in causal survival analysis to establish the identifiability of our parameter of interest, the counterfactual survival function $\psi^{a,t}$ at treatment $a$ and time step $t$, as defined in sub-section \ref{subsection:parameter}. Then, we show how using a first-order correction of an estimator of $\psi^{a,t}$ results in a meta estimator that includes the \textit{inverse probability weights} (IPW) defined therein and is closely related to the efficient influence function-based one-step estimator. In section \ref{sec:approach}, we introduce our approach starting with a new characterization of the IPW in the previous section by the Riesz representation theorem, which suggests a new procedure for estimating both the treatment propensity and at-risk probability. We then theoretically show the asymptotic efficiency of our method by showing that it is asymptotically linear with variance characterized by the efficient influence function. Finally, we provide simulations in section \ref{sec:experiment}

\section{Preliminaries}
\label{sec:prelim}

In this section, we review discrete and counterfactual survival analysis, discuss the identification assumptions required to identify the counterfactual survival curve using observable data, and introduce the essential concept of an efficient estimator in this context. All proofs in this section can be found in the Appendix \ref{sec:proof2122}.


\subsection{Discrete and Counterfactual Survival Analysis}
\label{sec:counterfactual_surv}
\paragraph{Discrete time.} In this paper, we work with a discrete time grid. Define $T \in \mathcal{T}$, where $\mathcal{T} = [t_{\max}] = \{0,1,\dots,t_{max}\}$, $t_{max} \in \mathcal{Z}_+$. $t_{max}$ may be chosen by the user, or as a result of administrative censoring. Denote $|\mathcal T|$ the total number of time points and $|t| = |\{u \in \mathcal T: u \le t\}|$ the number of time points less than or equal to $t$. We assume that $P(T=0)=0$ so that zero times are ruled out and the the survival curves always start from $(0,1)$. 

\sloppy\paragraph{Data structure.} We define the ideal data unit in counterfactual survival analysis as $(X, E, A, T(0), T(1), C(0), C(1))$, where: $X \in \mathcal{X} \subseteq \R^d$ is the covariate recorded prior to the beginning of the study; $A \in \{0, 1\}$ is a binary random variable indicating e.g. whether or not a patient receives the treatment; $T(a), C(a) \in \mathcal T$ are the counterfactual time-to-event (event time) and time-to-censoring (censor time) of interest under $A=a$. 
Define the factual event time $T = AT(1) + (1-A)T(0)$ and censor time $C = AC(1) + (1-A)C(0)$ and finally the observable time $\td T = \min\{T, C\}$ and event indicator $E = \one(T \le C)$. The observable data unit is now $O = (X, E, A, \td T)$. 

Define the marginal-hazard and marginal-survival functions of latent times $T$ and $C$ as:
\begin{equation}
\label{eq:marginal-functions}
\begin{aligned}
h_t(x,a) &= P(T=t | X=x,A=a,T\ge t)\\
S_t(x,a) &= P(T>t|X=x, A=a)\\
G_t(x,a) &= P(C>t|X=x, A=a).
\end{aligned}
\end{equation}
These definitions imply a one-to-one relationship between $h$ and $S$: letting $t-=\max(0,t-1)$, 
\begin{equation}
\label{eq:one-to-one}
\begin{aligned}
h_t(x,a) &= \frac{S_{t-}(x,a) - S_t(x,a)}{S_{t-}(x,a)}\\
S_t(x,a) &= \prod_{u \le t} (1-h_u(x,a)).
\end{aligned}
\end{equation}
Define the sub-hazard and sub-survival function of observable time and event indicator $\td T, E$ as:
\begin{equation}
\label{eq:sub-functions}
\begin{aligned}
\lambda_t(x,a) &= P(\td T = t, E=1| X=x, A=a,\td T \ge t)\\
H_t(x,a) &= P(\td T > t | X = x, A=a).
\end{aligned}
\end{equation}
$H_t$ is also called the \textit{at-risk} probability (at treatment $a$). When $T$ and $C$ are conditionally independent given $X$, the following lemma relates the sub-distributions and the marginal-distributions:
\begin{proposition}
If $T \perp C | X$ then: $h_t(x,a) = \lambda_t(x,a)$, therefore: 
\begin{equation}
S_t(x,a) = \prod_{u \le t}(1- h_u(x,a)) = \prod_{u \le t}(1- \lambda_u(x,a)).
\end{equation}

Additionally, the at-risk probability decomposes into a product of the marginal-survival functions of the event and censoring: $H_t(x) = S_t(x)G_t(x).$
\label{lem:random_censoring}
\end{proposition}
This equivalence enables the estimation of functions of latent time $T$ from observable time $\td T, E$. From here on, we will use the term hazard to refer to the sub-hazard. Finally, we also define the treatment propensity/probability as $\pi(X,a) = 1/ P(A=a|X)$.

\subsection{Causal parameter of interest and identification.} 
\label{subsection:parameter}
We will focus mainly on the \textit{counterfactual survival function at time $t \in \mc T$ and treatment $a \in \{0, 1\}$}: $\estimand^{a,t} = P(T(a) > t)$. Other commonly encountered parameters can be built from it such as the average survival effect at time $t$: $(\estimand^{1, t} - \estimand^{0, t})$ and the treatment-specific mean survival time $\sum_{t \le t_{max}} \estimand^{a,t}$ as well as its average effect counterpart.
Where convenient, we will write $\estimand$ in place of $\estimand^{a,t}$, letting
the treatment and time of interest be inferred from context.

\paragraph{Identification} To identify the counterfactual survival function using the observable data, similar to \citet{hubbard2000nonparametric, bai2013doubly,bai2017optimal,westling2023inference,diaz2019statistical, cai2020one}, we require the following testable and untestable assumptions:
\begin{itemize}
\item (A1) $T(a), C(a) \perp A \vert X$ for each $a \in \{0, 1\}$.
\item (A2) $T(a) \perp C(a) \vert A=a, X$ for each $a \in \{0, 1\}$.
\item (A3) $P(A=a|X) > 0$ almost surely.
\item (A4) $P(C(a) \geq t|X) > 0 $ positivity (censoring),
\end{itemize}
in addition to consistency and non-interference \cite{imbens2015causal}.
We provide explanations of these assumptions in the Appendix.
\begin{proposition}
When Assumptions~(A1)-(A4) hold, $\psi^{a,t}$ can be computed by observable quantities
\begin{equation}\label{eq:identification}
\begin{aligned}
\psi^{a,t} &= \E\left[S_t(X,a)\right] = \E\left[\prod_{u \le t}\left( 1- \lambda_u(X,a) \right)\right].
\end{aligned}
\end{equation}
\label{prop:identification}
\end{proposition}
Proposition ~\ref{prop:identification} tells us that we can estimate the causal parameter $\psi^{a,t}$ via the estimable parameters $S_t(X,a)$ or $\lambda_u(X,a)$ (one-to-one). The quality of our estimation of $\lambda_u(X,a)$ will therefore directly influence that of $\psi^{a,t}$, and it will be useful to look at $\psi^{a,t}$ as a \textit{functional} i.e. $\psi^{a,t}(f) = \E[\prod_{u\le t} (1 - f_u(X,a))]$.

As discussed in \citet[Section 5.3]{westling2023inference}, many causal survival effects can be then identified, e.g. the additive effect  $\psi^{1,t} - \psi^{0,t}$ and multiplicative effect $\psi^{1,t}/\psi^{0,t}$ at time $t$. Furthermore, asymptotic approximations for these effects can be derived from corresponding ones for the causal survival curve via the delta method.

\subsection{Efficient estimation of $\psi^{a,t}$ from first-order approximation}
Given an estimate $\hat \lambda$ of the hazard $\lambda$, which we can get e.g. by using a machine-learning method of our choice, the following `plug-in estimate' is a natural estimate of $\psi(\lambda)$. 
\begin{equation}
\psi_n(\hat\lambda) = \frac{1}{n}\sum_{i=1}^n \left[\prod_{u\le t} (1 - \hat \lambda_u(X_i,a)) \right]
\end{equation}

We can improve on this using a first-order correction. Consider the first-order Taylor expansion of $\psi$ around $\hat \lambda$. 
\begin{equation} 
\label{eq:expansion}
\psi(\hat\lambda + h) \approx \psi(\hat \lambda) + d\psi(\hat \lambda)(h)
\end{equation}
Here $d\psi(\hat\lambda)(h)$ is the derivative of $\psi$ at $\hat\lambda$ in the direction $h$. For the specific choice $\hat h = \lambda-\hat\lambda$, the error in our hazard estimate, we get a first-order corrected estimator of $\psi(\lambda)$. This derivative term $d\psi(\hat\lambda)(\lambda - \hat\lambda)$ is something we must estimate. In Lemma \ref{lemma:linearization} in the Appendix, we characterize it this way.
\begin{equation}
\label{eq:derivative-1}
\begin{aligned}
d\psi(\hat\lambda)(h) &= \E \left[\sum_{u \le t}\hat r_u(X,a) h_u(X,a) \right] \quad \text{ for } \\
\hat r_u(X,a) &= -\hat S_t(X,a)\frac{\hat S_{u-}(X,a)}{\hat S_u(X,a)}
\end{aligned}
\end{equation}
We cannot yet use an estimator in this form. While $\hat \lambda$ and $\hat r$ are known quantities derived from our hazard estimate, $\hat h = \lambda-\hat\lambda$ involves the true hazard as well. What is nice is that we have a natural proxy for $\lambda_u$, $Y^u = \one(E=1, \td T=u)$, in the sense that $\lambda$ is the proxy's conditional mean.
\begin{equation}
\lambda_u(x,a) = \E[Y^u \mid X=x, A=a, \tilde T \ge u].
\end{equation}
Now we'll discuss some technical details and notation.

\paragraph{Hazard-like functions.}
$\lambda-\hat\lambda$ is a function with a specific structure, so it is not necessary for the linear approximation \eqref{eq:expansion} to hold for all functions $h$. Instead, we will think of $d\psi(\hat\lambda)(\cdot)$ as a linear functional on the space of `hazard-like functions' $\Lambda$, which we will now define. To do this, it will be useful to work with an additional random variable $G^u = \one(\td T \ge u)$, which we will think of a coordinate of the \emph{random vector} $G = [G^1, G^2,\dots, G^{t_{max}}]$. We will call $G$ a history, as $G^u$ tells us whether our event has happened prior to $u$. In terms of this random vector, the hazard is $\lambda_u(x,a,1)$ where, abusing notation, we define
\begin{equation}
\begin{aligned}
\lambda_u(x,a,g) = \E[ Y^u \mid X=x, A=a, G^u = g^u] 
\end{aligned}
\label{eq:noisy-evaluation}
\end{equation}
The value of the hazard $\lambda_u$ at time $u$ depends on $g$ only through its $u$th $g^u$; our space of hazard-like functions $\Lambda$ is the set of functions with this property. Acknowledging this, we will use the notation $h_u(x,a,g)$ and $h_u(x,a,g^u)$ interchangeably for $h\in \Lambda$, including the use of $h_u(x,a,1)$ and $h_u(x,a,0)$ corresponding to the specific values $g^u=0$ and $g^u=1$. Note that in the notation of previous sections, the hazard corresponds to the specific value $\lambda(x,a,1)$, with $\lambda(x,a,0)=0$ for all $x$ and $a$.

\paragraph{Cross-fitting.}
To simplify our discussion, we will from this point assume that $\hat \lambda$ is \textit{cross-fit} i.e. fit on an auxiliary sample independent of the current sample \citep{chernozhukov2018double}, and therefore is a \textit{deterministic} function with respect to current sample. 
Where expectations are written, they are averages over the current sample, conditional on the one used to fit $\hat\lambda$.
Note that the survival curve can be computed from the hazard and vice-versa \ref{eq:one-to-one}, so $\hat S$ refers to the survival curve corresponding to $\lambda$.

\subsection{Inverse Probability Weighting (IPW)}

A difficult arises when we attempt to use the proxy $Y_u$ in place of $\lambda$ in \eqref{eq:derivative-2}. What we want there is $\lambda_u(X,a,1)$; the hazard for a specific level of treatment and history. What we have is a noisy version of the random variable $\lambda_u(X,A,G^u)$, the hazard at the naturally-occurring level of treatment and history. The solution to this is inverse probability weighting. In particular, we  rewrite the derivative as follows.
\begin{equation}
\begin{aligned}
d\psi(\hat\lambda)(h) &= \E \left[\sum_{u \le t}\hat r_u(X,a) h_u(X,a,1) \right]
=\sum_{u \le t}\E \Big[\gamma_u(X,A,G) h_u(X,A,G)\Big] \\
&\quad\text{for}\quad 
\gamma_u(X,A,G) = \hat r_u(X,a) \frac{\one(A=a, G^u=1)}{P(A=a,G^u=1|X)}.
\end{aligned}
\label{eq:derivative-2}
\end{equation}
The law of iterated expectations justifies the substitution of $Y_u$ for $\lambda$ in $d\psi(\hat\lambda)(\lambda-\hat\lambda)$
expressed in the form above.
\begin{equation}
d\psi(\hat \lambda)(\hat h) = \sum_{u \le t}\E \Big[\gamma_u(X,A,G) (Y^u - \hat \lambda_u(X,A,G))\Big]
\label{eq:derivative-3}
\end{equation}
Estimating these weights; substituting a sample average for this expectation; and plugging into our Taylor expansion \eqref{eq:expansion} yields the so-called `one-step' estimator. 
\begin{equation}
\label{eq:the-generic-estimator}
\begin{aligned}
&\hat \psi = \psi_n(\hat \lambda) 
+ \frac{1}{n}\sum_{i=1}^n \sum_{u \le t} \hat \gamma_u(X_i,A_i,G_i) \left\{Y_i^u - \hat \lambda_u(X_i,A_i,G_i)\right\}
\end{aligned}
\end{equation}

When $\lambda$ and $\gamma$ are estimated sufficiently well, this estimator will be asymptotically efficient i.e. $\hat \psi - \psi \rightarrow_d \frac{1}{n}\sum_{i=1}^n \phi^{a,t}(O_i) + o_p(n^{-1/2})$ where
\begin{equation} 
\label{EIF}
\begin{aligned}
\phi^{a,t}(O) &= S_t(X,a) - \psi^{a,t}(\lambda) + \sum_{u \le t}\gamma_u(X,A,G)\left\{Y^u - \lambda_u(X,A,G)\right\}
\end{aligned}
\end{equation}
That is, it is asymptotically equivalent to a version of the one-step estimator in which the true values of $\lambda$ and $\gamma$, not the estimated ones, are used. A sufficient condition for this is that the errors $\hat\lambda-\lambda$ and $\hat\gamma-\gamma$ converge at $o(n^{-1/4})$ rate. Roughly speaking, this is because the remainder of the Taylor approximation \eqref{eq:expansion} is quadratic.

Note that $\gamma_u$ contains the \textit{inverse probability weight} (IPW) $\frac{1}{P(A=a,G^u=1|X)}$ which is also the product of the inverse treatment probability $\frac{1}{\pi(X,a)}$ and the inverse at-risk probability $\frac{1}{H_{u-}(X,a)}$ at time $u$, both of which can be large. An intuitive explanation of the role of the IPW is that it re-weighs our observable data $(A,G^u)$ (up-weighing by $\gamma_u(X,a,1)$ when $A=a, G^u=1$ and setting to $0$ otherwise) so that we still get an unbiased estimate based on true hazard $\lambda(X,a,1)$ at potentially counterfactual value $(a,1)$ of $(A, G^u)$. 

The inverse probability weights $\gamma$ are usually estimated using regression-type losses of its nuisance components, the treatment $\pi(X,a)$ and at-risk $H_u(X,a)$ probability. However, this solution suffers from instability, common in all inverse weight-based estimators. When the ground truth functions $H_{u-}(X,a)$ and $\pi(X,a)$ are very small at some observations, naturally their estimators tends to be very small, hence, slight errors in estimating them can result in large errors in the estimation of their inverses (i.e. $1/(x-\epsilon)-1/x \approx \epsilon/x^2$). Thus, lack of overlap in the data means that such estimators will be unstable, with very large sample sizes needed to get estimates of $\pi$ and $H$ accurate enough to tolerate inversion. Moreover, even when overlap in the data is not poor, moderate-sized errors in the estimation of $\pi$ and $H$ that occur at smaller sample sizes results in similar issues.  It is common practice to clip these weights to a reasonable range before using them, but ad-hoc clipping often results in problematic levels of bias. In short, we lack practical and theoretically sound methods to make this inversion step work reliably. To avoid this problem, we propose an alternative covariate balancing approach below.



\section{Approach}
\label{sec:approach}
Another way of thinking about the inverse probability weights focuses on what they do rather than their functional form. We will work with an inner product on our space $\Lambda$.
\begin{equation}
\inner{f}{g} = \E \left[\sum_{u \le t} f_u(X,A,G^u) g_u(X,A,G^u)\right]
\end{equation}
for all $f, g \in \Lambda$. We observe that the Riesz representation theorem guarantees that, in the space $\Lambda$, there exists a \textbf{unique} element $\gamma$ that acts (via an inner product) on every function $h\in\Lambda$ like the functional derivative does:
\begin{equation}
\label{eq:riesz-rep}
\begin{aligned}
d\psi(\hat \lambda)(h) = \inner{\gamma}{h} = \sum_{u \le t}\E\left[ \gamma_u(X,A,G) h_u(X,A,G) \right] \quad \textbf{ for all} \quad h \in \Lambda
\end{aligned}
\end{equation}
and, in particular, taking $h=\lambda-\hat\lambda$, we have
\begin{equation}
\begin{aligned}
\inner{\gamma}{\lambda-\hat\lambda} &= \sum_{u \le t}  \E\left[\gamma_u(X,A,G) \left\{ Y^u - \hat \lambda_u(X,A,G) \right\} \right] 
\end{aligned}
\label{eq:riesz-rep-2}
\end{equation}
Looking at Equation~\ref{eq:derivative-2}, we see that the explicit solution to Equation~\ref{eq:riesz-rep} is, 
as promised, the inverse probability weighting function we called $\gamma$ in the previous section. What has changed is emphasis. Here we are characterizing $\gamma$ by \emph{what it does} instead of how it is calculated. This new perspective will is more alligned with how we will be estimating this function.

\subsection{Estimating the Riesz representer $\gamma$.}
Our approach avoids the problematic inversion in the one-step estimator by using the definition of the IPW as the Riesz representer in equation~\ref{eq:riesz-rep}. In particular, we consider the equivalence of the inner product $\smallinner{\gamma}{h}$ to the derivative evaluation $d\psi(\hat\lambda)(h)$, rather than its analytic form. 
First, inspired by the explicit characterization in \eqref{eq:derivative-2}, we ask that, for a set of functions $h\in\Lambda$, our weights $\hat \gamma$ satisfies:
\begin{equation} 
\label{eq:sample-balance}
\begin{aligned}
\frac{1}{n}\sum_{i=1}^n \sum_{u \le t} \hat r_u(X_i, a, 1) h_u(X_i,a,1) 
&\approx \frac{1}{n}\sum_{i=1}^n \sum_{u \le t} \hat \gamma_{iu} h_u(X_i,A_i,G_i).
\end{aligned}
\end{equation}
In the simpler setting in \citet{hirshberg2021augmented}, this approximation has meaning beyond that, suggested by its relationship to the population analog \eqref{eq:riesz-rep}. Along with the accuracy of the estimator $\hat\lambda$ and therefore of the linear approximation in \eqref{eq:derivative-1}, the quality of the \emph{in-sample approximation} described by \eqref{eq:sample-balance} for the specific function $\hat h=\lambda-\hat\lambda$ is one of two essential determinants of the estimator's bias. See Appendix~\ref{sec:sketch}, where we include an informative decomposition of our estimator's error and further discussion.

In light of that, we generalize the approach of \citet{hirshberg2021augmented} for estimating $\gamma$ by ensuring that \eqref{eq:sample-balance} holds for a set $\mathcal{M}$ of hazard-like functions $h$, which we deem as a \emph{model} for function $\hat h = \lambda-\hat\lambda$. In particular, we ask for weights that (i) are not too large, to control our estimator's variance, and (ii) ensure the approximation in Equation (\ref{eq:sample-balance}) is accurate uniformly over model $\mathcal{M}$. Thus, choosing a norm $\norm{\cdot}$ and 
taking the set of vector-valued functions 
$[h_1 \ldots h_{t_{max}}]$ with $\sum_u \norm{h_u}^2 \le 1$ as our model for $[\hat h_1 \ldots \hat h_{t_{max}}]$, we choose weights by solving the following optimization problem
\begin{equation}
\label{eq:weight-optimization}
\begin{aligned}
&\hat\gamma = \operatorname*{argmin}_{\gamma \in \R^{n|\mc T|}} \left\{ I(\gamma)^2 + \frac{\sigma^2}{n}\sum_{i=1}^n \sum_{u \le t} \gamma_{iu}^2 \right\} \quad \text{where}\\
&I(\gamma) = \max_{\substack{\sum_{u \le t}\norm{h_u}^2 \le 1 }} \Bigg\{\frac{1}{n}\sum_{i=1}^n \sum_{u \le t} \hat r_u(X_i, a, 1) h_u(X_i,a,1)\\
&\qquad\qquad\qquad\qquad - \ \frac{1}{n}\sum_{i=1}^n \sum_{u \le t}  \gamma_{iu} h_u(X_i,A_i,G_i).\Bigg\}
\end{aligned}
\end{equation}
What remains is to choose this norm with the intention that $\smallnorm{\hat h_u}$ is small for all $u$. If our model is correct in the sense that $\sum_u \smallnorm{\hat h_u}^2 \le B^2$, then $B$ times the maximal approximation error $I(\hat\gamma)$ bounds the approximation error in \eqref{eq:sample-balance}.

As usual, there is a natural trade-off in choosing this model---if we take it to be too small, $\smallnorm{\hat h_u}$ will be large or even infinite; on the other hand if we take it to be too large, we will be unable to find weights for which the approximation \eqref{eq:sample-balance} is highly accurate for all functions in the model. Choosing a norm with a unit ball that is a \emph{Donsker class}, e.g. a Reproducing Kernel Hilbert Space (RKHS) norm like the one we use in our experiments, is a reasonable trade-off that is common in the literature on minimax and augmented minimax estimation of treatment effects \citep[e.g.,][]{hirshberg2019minimax, kallus2016generalized}. When we do this, our estimator will be asymptotically efficient if the true hazard functions $\lambda_u$ are in this class and we estimate them via empirical risk minimization with appropriate regularization. We discuss the computational aspects of this problem in Appendix~\ref{sec:weight-optimization-details}.

\subsection{Asymptotic Efficiency}
In this section, we will discuss the asymptotic behavior of our estimator, giving sufficient conditions for it to be asymptotically efficient. Recall that we assume $\hat\lambda$ is cross-fit.

Our first condition is that it converges faster than fourth-root rate. This ensures the error of our first-order approximation \eqref{eq:expansion}, quadratic in the error $\hat h=\lambda-\hat\lambda$, is asymptotically negligible. 
\begin{assumption} for all $u \le t$,
\label{ass:rate}
$$\norm{\hat\lambda_u - \lambda_u}_{L_2(P)} = o_p(n^{-1/4}).$$
\end{assumption}
Our second condition is that its error is \emph{bounded} in the norm $\norm{\cdot}$ used to define our model. This ensures that the bound on $I(\hat\gamma)$ achieved via the optimization in \eqref{eq:weight-optimization} implies a comparable bound on the error of the derivative approximation $\smallinner{\hat\gamma}{h}_n\approx d\psi(\hat\lambda)(h)$ in \eqref{eq:sample-balance}
for the relevant perturbation $\hat h=\lambda-\hat\lambda$.
\begin{assumption} for all $u \le t$, 
\label{ass:strongnorm-bounded}
$$\norm{\hat \lambda_u - \lambda_u} = O_p(1).$$
\end{assumption}
If our model is correctly specified in the sense that $\norm{\lambda_u} < \infty$ for $\forall u$,
a sensibly tuned $\norm{\cdot}$-penalized estimator of $\lambda_u$ will have these two properties \citep[see, e.g.,][Remark 2]{hirshberg2021augmented}.

Our third condition is that we have a sufficient degree of overlap. 
\begin{assumption}
\label{ass:overlap}
$\E[1/P(A=a, \td T \ge u \mid X)] < \infty$
\end{assumption}
This is a substantially weakened version of the often-assumed `strong overlap' condition that $P(A=a, \td T \ge u \mid X)$ is bounded away from zero,
allowing this probability to approach zero for some $X$ as long as it is `typically' elsewhere. 

Our final condition is, for the most part, a constraint on the complexity of our model.
\begin{assumption}
\label{ass:donskerity}
The unit ball $\mathcal{B}=\left\{ h \ : \ \norm{h} \le 1 \right\}$ is Donsker and uniformly bounded in the sense that $\max_{h : \norm{h} \le 1}\norm{h}_{\infty} < \infty$.
Furthermore, $\left\{\frac{1(A=a, \td T \ge u) \ h(\cdot)}{P(A=a, \td T \ge u \mid X=\cdot)} \ : \ \norm{h} \le 1 \right\}$ is Donsker for all $u$.
\end{assumption}
The `furthermore' clause here is implied by the first clause if strong overlap holds, as multiplication by 
the inverse probability weight $1/P(A=a, \td T \ge u \mid X=\cdot)$ will not problematically
increase the complexity of the set of functions $h \in \mathcal{B}$ \emph{if those weights are bounded}.
It is also implied by a strengthened version of the first clause, the condition that $\mathcal{B}$ have a bounded uniform entropy integral,
with our weaker overlap requirement of Assumption~\eqref{ass:overlap} \citep[Remark 4]{hirshberg2021augmented}. 
This tends not to be a problematic strengthening, as the unit balls of many norms that are used in this context,
including H\"older and Sobolev norms involving bounds on $s > \operatorname{dimension}(X)/2$ derivatives
and the RKHS norm we use in our experiments, do have bounded uniform entropy integrals \citep{nickl2007bracketing, zhou2002covering}.

We are now ready to state our main theoretical result:
\begin{theorem}
\label{thm:efficiency}
Suppose $\hat\lambda$ is a hazard estimator fit on an auxiliary sample and Assumptions~\ref{ass:rate}-\ref{ass:donskerity} are satisfied. 
Then the estimator $\hat\psi^{a,t}$ described in \eqref{eq:the-generic-estimator}, using the Riesz Representer estimate $\hat\gamma$ obtained by solving 
\eqref{eq:weight-optimization} for any fixed $\sigma > 0$, is \emph{asymptotically linear} with influence function $\phi^{a,t}$. That is, it has the asymptotic approximation
$$\hat\psi^{a,t} - \psi^{a,t} = \frac{1}{n}\sum_{i=1}^n \phi^{a,t}(O_i) + o_p(n^{-1/2}).$$
where $\phi^{a,t}$ is the efficient influence function defined in Equation \ref{EIF}.
\end{theorem}
 It follows, via the central limit theorem, that under these conditions $\sqrt{n}(\hat\psi - \psi)$ is asymptotically normal with 
mean zero and variance $V=\E\left[ \phi^{a,t}(O)^2\right]$. This justifies a standard approach to inference based on this
asymptotic approximation, i.e., based on the t-statistic $\smash{\sqrt{n}(\hat\psi-\psi)/\hat V^{1/2}}$ being approximately standard normal if 
$\hat V$ is a consistent estimator of this variance $V$.

As usual for estimators involving cross-fitting, 
working with multiple folds and averaging will yield an estimator 
with the same characterization without an auxiliary sample \citep[e.g.,][]{chernozhukov2018double}. The resulting
estimator is asymptotically linear on the whole sample and, having the efficient influence function $\phi$,
is \textit{asymptotically efficient}.\footnote{There are many equivalent formal descriptions of what asymptotic efficiency means \citep[e.g.,in][Chapter 25]{van2000asymptotic}. In essence, it implies that no estimator can be reliably perform better in asymptotic terms, either in terms of criteria for point estimation like mean squared error 
or in terms of inferential behavior like the power of tests against local alternatives.}

\newpage

\section{Experiments}

\label{sec:experiment}

\begin{figure*}[htb!]
\includegraphics[width=\textwidth]{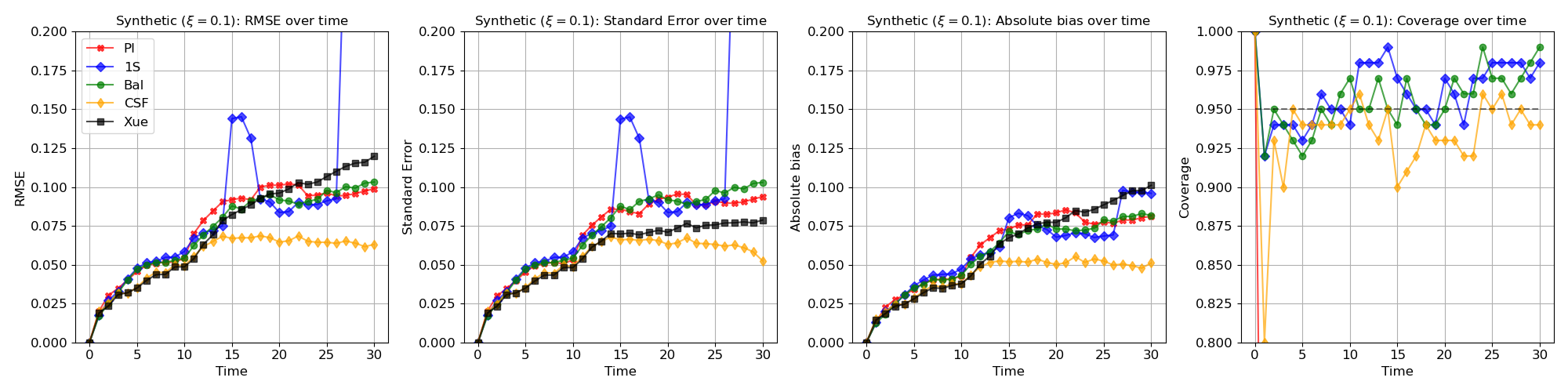} 
\includegraphics[width=\textwidth]{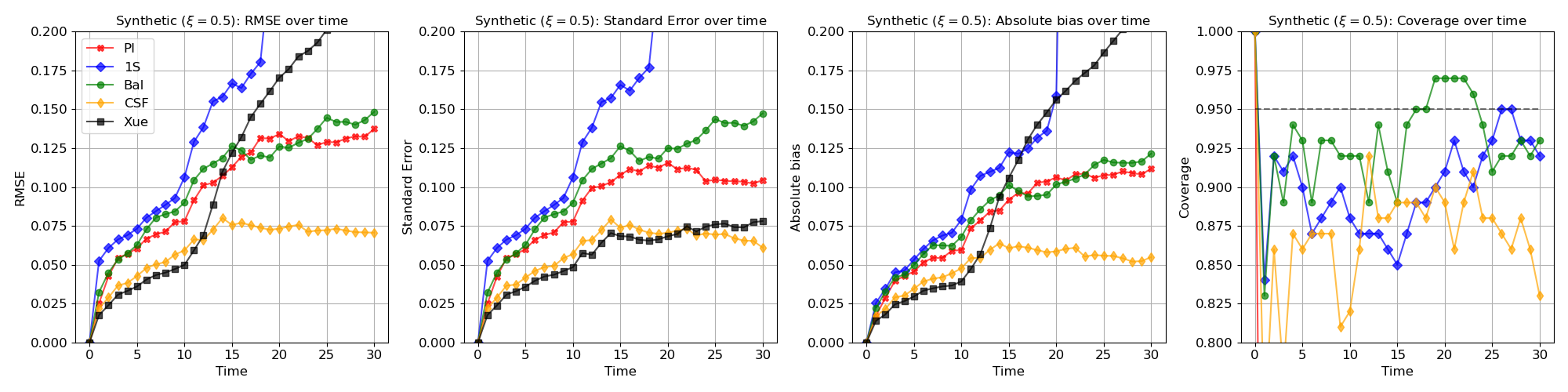} 
\includegraphics[width=\textwidth]{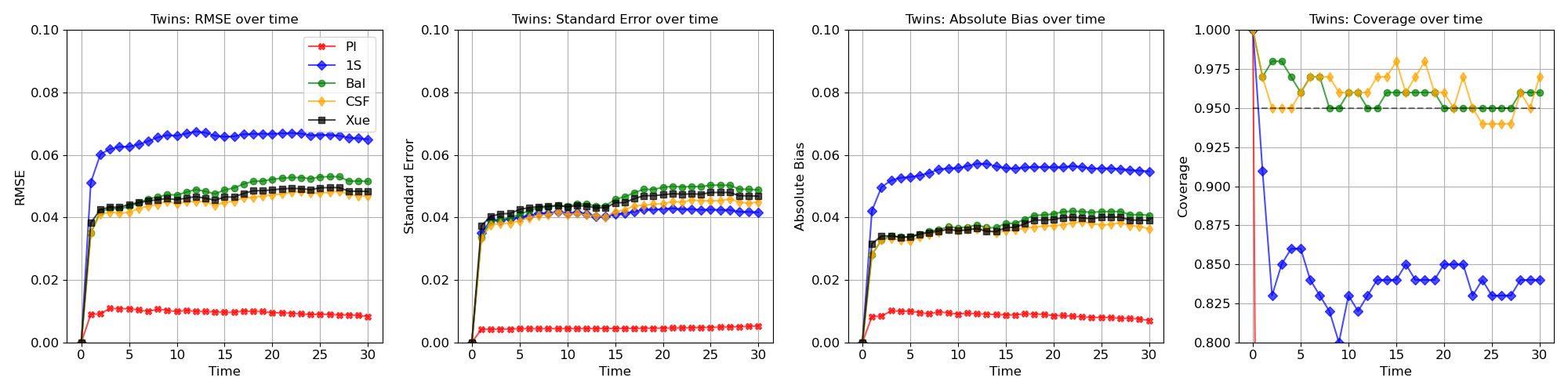} 
\caption{Metrics over time. The columns are RMSE, Standard Error, Absolute Bias and Coverage, while the rows are Synthetic (more overlap $\xi=0.1$), Synthetic (less overlap $\xi=0.5$) and Twins. The first 3 metrics generally increase because of the lack of overlap in time caused by a decreasing at-risk probability.}
\label{fig:metrics-over-time}
\end{figure*}

We now describe the experimental evidence concerning our estimator. As ground truth values are almost never  available in causal inference, synthetic or semi-synthetic data (i.e. synthesized data based on real data) is the standard. We focus on the experiments that provide the most insight into the behavior of our estimator and its competitors; additional experiments and metrics along with implementation details are included in Appendix \ref{sec:experiments_details}.

We primarily compare the performance of our balancing estimator (\textbf{Bal}) in estimating the additive effect at time t: $\Delta^t = \psi^{1,t} - \psi^{0,t}$ with other kernel-based approaches: the plugin estimator (\textbf{PI}), the one-step estimator (\textbf{1S}), and the approach of \citet{xue2023rkhs} (\textbf{Xue}). We also compare with the state-of-the-art causal survival forest (\textbf{CSF}). Our error metrics are RMSE, standard error, bias, and coverage of the 95\% confidence interval. 
We excluded \textbf{PI} and \textbf{Xue} from the coverage plots because there is no standard approach or theory supporting the construction of confidence intervals based on this estimators. From what we see, it appears inference based on these approaches (without bias correction) would be challenging because their bias tends to be large relative to their standard error.




We use the two datasets from \citet{curth2021survite} with slight modifications. The first dataset \textbf{Synthetic} has the following data generating process: covariate $X \sim \mc N(0, 0.8 \times I_{10} + 0.2 \times J_{10})$ where $I_{10}$ and $J_{10}$ are 10-dimensional identity matrix and all-one matrix respectively, treatment assignment $A \sim \text{Bern}\left(\xi \times \sigma\left(\sum_p x_p\right)\right)$ where $\sigma$ is a sigmoid function, and $\sum_p x_p$ is the sum of all covariate of a specific $x$. The parameter $\xi$ represents 'lack of overlap': the higher $\xi$, the more uneven propensity is. $\xi$ is 0.3 by default and varies in $\xi$ in $\{0.1, 0.2, 0.3, 0.4, 0.5\}$. We set $t_{max}=30$. The event and censor hazards are:
\begin{align*}
&h_t(a, x)=\\
&\begin{cases}
  0.1\sigma(-5x_1^2 - a \times (\one(x_3 \ge 0) + 0.5)) & \text{for}\ t \le 10 \\
  0.1\sigma(10x_2 - a \times (\one(x_3 \ge 0) + 0.5)) & \text{for}\ t > 10\\
\end{cases}\\
&h_{C,t}(a,x)=
\begin{cases}
(0.01 \times \sigma(10x_4^2) & \text{for}\ t < 30\\
1 & \text{for}\ t \ge 30\\
\end{cases}
\end{align*}
The second dataset is semi-synthetic based on the \textbf{Twins} dataset \citep{louizos2017causal, yoon2018ganite}. The time-to-event outcome is the time-to-mortality of each twin. Each observation $x$ has 30 covariates. The time unit is day and we also set $t_{max}=30$. We create artificial treatment and censoring: treatment assignment 
$A \sim \text{Bern}(\sigma(w_1^\top x + e))$ where $w_1 \sim \text{Uniform}(-0.1, 0.1)^{30\times 1}$ and $e \sim \mc N(0, 1^2)$. Censoring time $C \sim \text{Exp}(10\times \sigma(w_2^\top x))$ where $w_2 \sim \mc N(0, 1^2)$. The treatment is which twin is born heavier. We standardize covariate $x$ for training and only after creating the datasets. For each dataset, we sample 200 observations in each of $Q=100$ simulation runs.

Figure \ref{fig:metrics-over-time} shows the effect of lack of overlap over time. The first 3 metrics RMSE, Standard Error and Bias generally increase over time. \textbf{1S} consistently encounters large jumps because of the inversion of small estimated probability, while \textbf{Bal} does not. Compared to the other 3 baselines, we see that \textbf{Xue} performs badly as time increases because it assumes unconditionally independent censoring and so does not correct for censoring bias. \textbf{PI} performs on par with \textbf{Bal}, but note again that neither \textbf{Xue} and \textbf{PI} supports constructing valid confidence interval. The method that performs best is \textbf{CSF} which uses random forest and therefore is generally better than kernels, especially when the data consists of discrete jumps. Surprisingly, \textbf{CSF}'s first 3 metrics are quite stable from time 15 onwards, even across the 2 propensity overlap settings, while our method \textbf{Bal} increases more relatively. Regarding coverage, in \textbf{Synthetic}, \textbf{CSF} is slightly worse than \textbf{Bal} and \textbf{1S} throughout time, and in \textbf{Twins} performs similarly to \textbf{Bal}, while \textbf{1S} is significantly worse.

We remark that, except for \textbf{CSF}, all methods that we implement use kernel regressions. \textbf{CSF} also provides theoretical guarantees (asymptotic efficiency) similar to ours. It does not seem to suffer from small probability inversion; indeed random forests generally estimate more conservative probabilities even if the true values are more extreme. Our method \textbf{Bal} performs on par or better than other kernel-based methods and is particular stable compared to \textbf{1S}, showing that it significantly improves on the \textbf{1S} baseline whose problem we specifically target. 

\begin{figure*}[htb!]
\includegraphics[width=\textwidth]{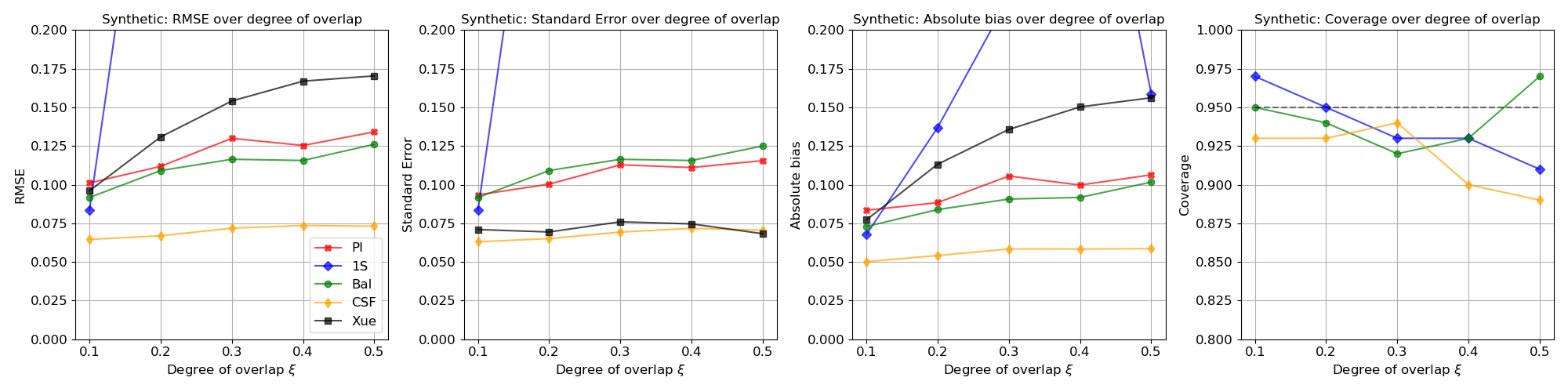} 
\caption{Metrics over different degree of lack of propensity overlap $\xi$ for \textbf{Synthetic}.}
\label{fig:metrics-over-propensity}
\end{figure*}

In figure \ref{fig:metrics-over-propensity}, we further plot the same metrics across different degree of propensity overlap $\xi$ at time $T=20$, where the larger $\xi$ is the less overlap. Again, we see that RMSE, Standard Error and Bias generally increase and coverage decreases with less overlap. We can make similar observations to the lack of overlap in time plots: \textbf{1S} is very unstable with large jumps, while \textbf{Bal} performs competitively in terms of RMSE, Standard Error and Bias, and performs best in terms of coverage at less overlap.

\section{Conclusion}
We propose a kernel covarate-balancing estimator for the counterfactual survival curve, an important statistic in causal inference and survival analysis. Theoretically, we show that it is asymptotically efficient. Empirically, the method improves over kernel-based baselines, especially the one-step estimator which is also asymptotically efficient but uses small estimated probability inversion. We found surprisingly that the state-of-the-art method causal survival forest performs very well even when the lack of overlap means the true treatment and at-risk probabilities are very small. In future works, we would like to investigate alternative covariate-balancing methods that are less dependent on kernels but still satisfy the asymptotic efficiency property. Neural networks as nuisance estimators for example can be extremely powerful because they can handle diverse and rich data types such as image, text and graph that tree-based methods are usually worse at, but on the other hands neural networks are even more prone to overfitting to extreme probability (close to 1 or 0). Therefore, combining our meta approach with neural networks can unlease the full potential of the benefit that covariate-balancing methods offer, while making our current theory work for the neural network estimators will be challenging and interesting.

\clearpage
\printbibliography
\newpage
\appendix
\onecolumn
\section{Description of implementation, datasets and additional results}
\label{sec:experiments_details}
\subsection{Implementation}
\label{sec:implementation}
Throughout, we use the RBF kernel with length scale 10. For all methods considered, we need to train a hazard estimator for time-to-event (the survival function can then be constructed from the hazard). We use the discrete \textit{logistic-hazard} model \citep{kvamme2019continuous} with the mean negative log-likelihood loss parameterized by the hazard function as:
\begin{equation*}
L(\{O\}_{i=1}^n; \lambda) = -\frac{1}{n}\sum_{i=1}^{n}\sum_{u \le T_i}[Y_u \log \lambda_u(X_i, A_i) + (1 - Y_u)\log (1 - \lambda_u(X_i, A_i))]
\end{equation*}
where $Y_u = \one(E_i=1, T_i=u)$. This loss breaks down into $|\mc T| \times 2$ independent binary cross-entropy losses for each $u \in \mc T$ and $a \in \{0, 1\}$: 
\begin{equation*}
\begin{aligned}
&L(\{O\}_{i=1}^n; \lambda) = \sum_{u \in \mc T} \sum_{a \in \{0, 1\}} L_u(\{O\}_{i=1}^n; \lambda_u(\cdot, a))
\end{aligned}
\end{equation*}
where
\begin{equation*}
 L_u(\{O\}_{i=1}^n; \lambda_u(\cdot, a)) \\= -\frac{1}{n}\sum_{i=1}^{n}\one(u \le T_i, A_i=a)[Y_u \log \lambda_u(X_i, a) \\+ (1 - Y_u)\log (1 - \lambda_u(X_i, a))] 
\end{equation*}

Therefore, we can fit $|\mc T| \times 2$ independent hazard models using kernel logistic regression.

\paragraph{Covariate-balancing}
The hazard estimate must be fit on a separate split of the data, therefore we divide the data into 2 folds. For each fold, we fit the hazard estimator on the other fold, then estimate the Riesz representer using the just obtained hazard estimate on the canonical fold. The result is one causal parameter estimate of the canonical fold, by solving the optimization problem in (\ref{eq:weight-optimization}). This gives us the time-to-event hazard and the Riesz estimates which we use to estimate the causal parameter of this canonical fold. Lastly, we obtain the final estimate by averaging the estimates across 2 folds.

\paragraph{Double robust estimation} We need to train a hazard estimator for time-to-censoring, which can be done similarly to the time-to-event case (but with events flipped), and a propensity estimator, for which we use a simple linear logistic regression (correctly specified in both datasets). We use cross-fitting \citep{chernozhukov2018double}, where we randomly divide the dataset into K folds (we used K=5 in our experiments). For each fold, we fit the time-to-event/censoring hazard and propensity estimators on the remaining folds and obtain their estimates on the canonical fold. Using the time-to-censoring and the propensity estimates, we obtain the Riesz estimate using \eqref{eq:weight-optimization}. Lastly we obtain the causal parameter estimate using \eqref{eq:the-generic-estimator}.

\subsection{The weight optimization problem}
\label{sec:weight-optimization-details}

We start by observing that we do not need to solve the optimization \eqref{eq:weight-optimization} all at once. It decomposes into
$t$ separate optimizations over timestep-specific weights $\gamma_{\cdot u}$ specific to a single timestep. 
\begin{lemma}
\label{thm:optimization-decomposition}
If the weights $\hat\gamma$ solve \eqref{eq:weight-optimization} then 
the $u$-timestep $\hat\gamma_{u}$ also solves the following $u$-timestep problem: 
\begin{equation}
\label{eq:u-timestep-problem}
\begin{aligned}
\hat \gamma_u &= \argmin_{\gamma_u \in \R^n} \left\{ I_u(\gamma_u)^2 + \frac{\sigma^2}{n} \sum_{i=1}^n \gamma_{iu}^2 \right\} \quad \text{where}\\
I_u(\gamma_u) &= \max_{\substack{\norm{h_u} \le 1}} 
\frac{1}{n}\sum_{i=1}^n \Big\{\hat r_{u}(X_i, a) h_{u}(X_i,a,1) - \gamma_{iu} h_{u}(X_i,A_i,G_i)\Big\}
\end{aligned}
\end{equation}
\end{lemma}

In the case that $\norm{\cdot}$ is the norm of an RKHS, we can use the representer theorem to further simplify this optimization. The representer theorem implies that it is sufficient to maximize over $h_u$ that can be written as $\sum_{j=1}^{2n} \alpha_j k(\cdot, Z_j)$ where $k$ is our space's kernel and $Z_j \in \{(X_1,A_1,G_1),\dots,(X_n,A_n,G_n),\dots,(X_1,a,1),\dots(X_n,a,1)\}$ the concatenation of $n$ observable data with natural values of $(A_i, G_i)$ and $n$ ideal data with counterfactual values $(a, 1)$. Making this substitution, we get the following characterization in terms of the $2n \times 2n$ kernel matrix $K$ satisfying $K_{ij} = k(Z_i,Z_j)$.
\begin{equation*}
    I_u(\gamma_u) = \max_{\alpha^T K \alpha \le 1} \left\{ \alpha^T K v_u \right\}
\end{equation*}
where $v_u = [-\gamma_u^\top; \hat r_u^\top]^\top$ where $\gamma_u, \hat r_u$ are vectors in $\R^n$ and $\hat r_{iu} = \hat r_u(X_i, a)$ and $[;]$ is vector concatenation. The maximum is achieved at $\alpha = v_u / \|K^{1/2}v_u\|$ and $I_u(\gamma_u)$ is:
\begin{equation*}
I_u(\gamma_u) = \|K^{1/2} v_u\|
\end{equation*}
Replacing into the outer minimization problem:
\begin{align*}
\hat \gamma_u &= \argmin_{\gamma_u \in \R^n} \Big\{ v_u^T K v_u
+ \frac{\sigma^2}{n} \sum_{i=1}^n \gamma_{iu}^2 \Big\}
\end{align*}
This quadratic-programming problem can be solved efficiently by most convex solvers, in particular we chose \texttt{cvxpy}. Now we show why we can decompose the original problem this way.

\begin{proof}[Proof of lemma \ref{thm:optimization-decomposition}] 
Let $L_u(\gamma_u, h_u)$ be defined as follows.
\begin{equation*}
    L_u(\gamma_u, h_u)= \frac{1}{n}\sum_{i=1}^n \Big\{\hat r_{iu} h_{u}(X_i,a,1) - \gamma_{iu} h_{u}(X_i,A_i,G_i)\Big\}
\end{equation*}
In terms of this function,
\begin{align*}
I(\gamma) &\stackrel{def}{=} \max_{\sum_u \smallnorm{h_u}^2 \le 1} \sum_{u} L_u(\gamma_u, h_u)\\
&=\max_{\substack{\beta \in \R^t \\ \sum_{u}\beta_u^2 \le 1}} \sum_{u} \max_{\|h_u\| \le \beta_u} L_u(\gamma_u, h_u) \\
&=\max_{\sum_{u}\beta_u^2 \le 1} \sum_{u}  \max_{\|h_u\| \le 1} L_u(\gamma_u, \beta_u h_u)\\
&=\max_{\sum_{u}\beta_u^2 \le 1} \sum_{u}  \beta_u \max_{\|h_u\| \le 1} L_u(\gamma_u, h_u)\\
&=\max_{\sum_{u}\beta_u^2 \le 1} \sum_{u}\beta_u  I_u(\gamma_u)\\
&= \sqrt{\sum_u I_u^2(\gamma_u)}.
\end{align*}
The last equality is due to Cauchy-Schwarz inequality, with equality achievable by setting $\beta_u$ proportional to $I_u(\gamma_u)$. 
If we substitute this expression into our original optimization \eqref{eq:weight-optimization}, we get:
{\allowdisplaybreaks\begin{equation*}
\begin{aligned}
\hat\gamma &= \operatorname*{argmin}_{\gamma \in \R^{n|\mc T|}} \Big\{ \sum_{u \le t} I_u(\gamma_u)^2 + \frac{\sigma^2}{n}\sum_{i=1}^n \sum_{u \le t}\gamma_{iu}^2 \Big\}\\
&= \operatorname*{argmin}_{\gamma \in \R^{n|\mc T|}} \Big\{ \sum_{u \le t} \Big( I_u(\gamma_u)^2 + \frac{\sigma^2}{n}\sum_{i=1}^n \gamma_{iu}^2 \Big) \Big\} \\
&= \sum_{u \le t} \operatorname*{argmin}_{\gamma \in \R^{n|\mc T|}} \Big( I_u(\gamma_u)^2 + \frac{\sigma^2}{n}\sum_{i=1}^n \gamma_{iu}^2 \Big)\\
\end{aligned}
\end{equation*}}
as each term in the sum over $u$ is a function of an individual $\gamma_u$, we can therefore solve a separate minimization sub-problem for each $u$:
\begin{equation*}
\hat \gamma_u = \argmin_{\gamma_u \in \R^n} \left\{ I_u(\gamma_u)^2 
+ \frac{\sigma^2}{n} \sum_{i=1}^n \gamma_{iu}^2 \right\}
\end{equation*}
proving our claim.
\end{proof}

\begin{remark}
If we choose the kernel $k$ to be the product kernel $k((x_1,a_1,g_1), (x_2,a_2,g_2)) = k_x(x_1, x_2) \one(a_1=a_2) \one(g_1=g_2)$, then it can be shown that $\gamma_{iu} = 0$ whenever $\one(A_i= a, G^u_i=1) = 0$, and we need only model $h(x)$ instead of $h(x,a,g)$. Therefore problem \ref{eq:u-timestep-problem} is also equivalent to the following problem with an optimal solution $\hat \omega$ that satisfies $\hat \gamma_{iu} = \hat r_{iu} \one(A_i=a, G^u_i=1)\hat \omega_{iu}$
\begin{equation}
\begin{aligned}
\hat \omega_u &= \argmin_{\omega_u \in \R^n} \left\{ (\hat r_{u}(1 - \mathcal{I}_u \odot \omega_u))^\top K_x (\hat r_{u}(1 - \mathcal{I}_u \odot \omega_u)) + \frac{\sigma^2}{n} \sum_{i=1}^n \one(A_i=a, G^u_i=1)\omega_{iu}^2 \right\}
\end{aligned}
\end{equation}
where $K_{x,ij} = k_x(X_i, X_j)$ the kernel between covariates only, $\mathcal{I}_{iu} = \one(A_i=a, G^u_i=1)$ and $\odot$ is the element-wise product. This formulation used by most covariate-balancing works e.g. \citet{kallus2022optimal, ben2021balancing} is conceptually simpler and slightly easier to implement, and it shows explicitly that we are optimizing for the inverse probability weight by $\omega$.
\end{remark}

\newpage
\section{Identification and Proof of Supporting Lemmas}
\label{sec:proof2122}

\begin{proof}[Proof of Lemma \ref{lem:random_censoring}]
\begin{align*}
\lambda_t(x) &= P(E=1, \td T=t | X=x) / P(\td T \ge t | X=x)\\
&= P(T=t, C\ge t|X=x) / P(\td T \ge t | X=x)\\
&= P(T=t|X=x) P(C\ge t|X=x) / P(\td T \ge t|X=x)\\
&= \frac{P(T=t|X=x)}{P(T\ge t|X=x)} \frac{P(T\ge t|X=x)P(C\ge t|X=x)}{P(\td T \ge t|X=x)}\\
&= h_t(x)
\end{align*}
and $H_t(x) = P(\td T > t|X=x) = P(T > t, C > t|X=x) = P(T> t|X=x)P(C> t|X=x) = S_t(x)G_t(x)$
\end{proof}

\paragraph{Identification} To identify the counterfactual survival function using the observable data, similar to \citet{hubbard2000nonparametric, bai2013doubly,bai2017optimal,westling2023inference,diaz2019statistical, cai2020one}, we require the following testable and untestable assumptions: 
\begin{enumerate}
\item (A1) $T(a), C(a) \perp A \vert X$ for each $a \in \{0, 1\}$.
\item (A2) $T(a) \perp C(a) \vert A=a, X$ for each $a \in \{0, 1\}$.
\item (A3) $P(A=a|X) > 0$ almost surely.
\item (A4) $P(C(a) \geq t |X) > 0 $ positivity (censoring),
\end{enumerate}

in addition to consistency and non-interference \cite{imbens2015causal}. (A1), also referred to as no unmeasured confounders, selection on observables, exogeneity, and conditional independence, asserts that the potential outcomes/potential censoring times, and treatment assignment are independent given confounders. This assumption implies that all relevant information regarding treatment assignment and follow-up censoring times is captured in the available data. (A2) states that the potential follow-up and censoring times are independent given the treatment and confounders. (A3) and (A4) assert that the probability of receiving or not receiving treatment, as well as the probability of censoring, given confounders, is greater than zero. 

\begin{proof}[Proof of Proposition \ref{prop:identification}]
\begin{align*}
P(T(a) > t) &= \E[P(T(a) > t) | X]\\
&= \E[P(T > t|X,A=a)]\\
&= \E[S_t(X,a)]\\
&= \E\left[\prod_{u \le t}\left(1- h_u(X,a) \right)\right]\\
&= \E\left[\prod_{u \le t}\left( 1- \lambda_u(X,a) \right)\right]\\
\end{align*}
where we used iterated expectation in the first equation, (A1) in the second equation, (A2) in the forth equation, and Lemma~\ref{lem:random_censoring} in the last equation.
\end{proof}
\newpage

\section{Proof of theorem \ref{thm:efficiency}}
We first recall and introduce additional notations for this section. 
Denote $Z = (X, A, G)$ the conditioning structure of the hazard function $\lambda$, and $Y^u = \one(E = 1, \td T = u)$. 
As $r$ and $S$ are functions of $\lambda$, we use $\hat r$ and $\hat S$ as functions of $\hat \lambda$ to denote the estimator counterparts.
We drop the superscripts $a, t$ of $\psi^{a,t}$ since they are not relevant in the proof.  
We use $\psi = \psi(\lambda)$ for the parameter of interest $\E[S_t(X, a)]$, $\psi_n(\lambda)$ for the sample analog $\frac1n\sum_{i=1}^n S_t(X_i, a)$. As $S$ and $\lambda$ are one-to-one, $\psi(\hat \lambda)$ would be $\E[\hat S_t(X,a)]$. We use $\hat\psi$ to denote our estimator, i.e.,
\begin{equation}
\hat \psi = \psi_n(\hat\lambda) + \frac{1}{n}\sum_{i=1}^n  \sum_{u\le t} \hat\gamma_{iu} \left\{ Y^u_i - \hat\lambda_u(Z_i) \right\}.
\end{equation}

Let's start by recalling what we are proving
\begin{equation}
\label{eq:claim-appendix}
\hat\psi - \psi = \frac{1}{n}\sum_{i=1}^n \phi(\lambda)(O_i) + o_p(n^{-1/2}) 
\end{equation}
where
\begin{align*}
\phi(\lambda)(O) &= \left\{S_t(X,a) - \psi\right\} + \sum_{u\le t}\gamma_u(Z) \left\{Y^u - \lambda_u(Z)\right\}.
\end{align*}
where $\gamma_u(Z) = \one(A=a, G^u=1)r_u(X,a) \omega(X,a)$ and $\omega_u(X,a) = \frac{1}{H_{u-}(X,a) \pi(X,a)}$ and $r_u(X,a) = - S_t(X,a)\frac{S_{u-}(X,a)}{S_u(X,a)}$. To do this, we will work with this error decomposition.
\begin{equation}
\begin{aligned} 
\hat \psi - \psi(\lambda)
&= \psi_n(\hat\lambda) + \frac{1}{n}\sum_{i=1}^n  \sum_{u\le t} \hat\gamma_{iu} \{ Y_{iu} - \hat\lambda_u(Z_i)\} - \psi(\lambda) 
\\
&= \left\{\psi_n(\hat\lambda) - \psi(\hat\lambda)\right\} + \frac{1}{n}\sum_{i=1}^n  \sum_{u\le t} \hat\gamma_{iu}\Big[ \left\{ Y_{iu} - \hat\lambda_u(Z_i) \right\} -  d\psi(\hat\lambda)(\lambda - \hat\lambda)\Big] \\
& \qquad + \left\{ \psi(\hat\lambda) + d\psi(\hat\lambda)(\lambda - \hat\lambda) - \psi(\lambda) \right\}.
\end{aligned}
\end{equation}

We prove \eqref{eq:claim-appendix} in three steps. Throughout, we will work conditionally on the auxilliary sample used to estimate $\hat \lambda$, so we can act as if it is a deterministic function. This will imply that our claim holds where $o_p$ refers to probability conditional on the auxilliary sample
and therefore also that it holds where $o_p$ refers to unconditional probability. 

\paragraph{Step 1.} 
The third term in this decomposition, the error of our linearization of $\psi$ around $\hat\lambda$,
is $o_p(n^{-1/2})$. Lemma~\ref{lemma:linearization} below shows that this is implied by our assumption that
$\hat \lambda$ converges $\lambda$ at faster-than-fourth-root rate.

\paragraph{Step 2.}
Lemma~\ref{lemma:negligible-imbalance} below concludes that the second term in our decomposition has the following asymptotic approximation.
\begin{equation*}
\frac{1}{n}\sum_{i=1}^n  \sum_{u \le t}\Big[ \gamma_u(Z_i) \{ Y^u_i - \lambda_u(Z_i) \}\Big] + o_p(n^{-1/2})
\end{equation*}
\paragraph{Step 3.}
The sum of the first term in our decomposition and the non-negligible part of \textbf{Step 2} is 
\begin{equation}
\begin{aligned}
&\psi_n(\hat \lambda) - \psi(\hat \lambda) + \frac{1}{n}\sum_{i=1}^n \sum_{u \le t} \Big[\gamma_u(Z_i)
\{ Y^u_i - \lambda_u(Z_i) \} \Big] \stackrel{def}{=}\frac{1}{n}\sum_{i=1}^n\hat\phi(O_i) 
\end{aligned}
\end{equation}
To complete our proof of the claim \eqref{eq:claim-appendix}, we show 
\[ \frac{1}{n}\sum_{i=1}^n \left\{\hat\phi(O_i) - \phi(O_i) \right\} = o_p(n^{1/2}). \] 
Because  this is an average of independent and identically distributed terms with mean zero, its mean square is $1/n$ times the variance of an individual term; thus, all we have to do is show that the variance of $(\hat\phi-\phi)(O_i)$ goes to zero. In Lemma~\ref{lemma:eif-continuity} below, we show that this is a consequence of the convergence of $\hat\lambda \to \lambda$.

We conclude by stating and proving our lemmas.

\begin{lemma} For all $t \in \mc T$,
\begin{align*}
&\hat S_t(X,a) - S_t(X,a) = \sum_{u \le t}\hat S_t(X,a) \frac{S_{u-}(X,a)}{\hat S_{u-}(X,a)}\frac{\hat S_{u-}(X,a)}{\hat S_u(X,a)}\left(\lambda_u(X,a) - \hat \lambda_u(X,a)\right).
\end{align*}
Furthermore,
\[
\norm{S_t(X,a) - \hat S_t(X,a)}_{L_2(P)} \le |\mc T| \max_{u \le t}\norm{\lambda_u(X,a) - \hat \lambda_u(X,a)}_{L_2(P)}
\]
\label{lem:rate_S}
\end{lemma}

\begin{lemma} 
\label{lemma:linearization}
Let $\hat \lambda$ and $\lambda$ be two hazards and $\hat S$ and $S$ the associated survival functions. Then the functional $\psi(h)$ evaluated at $h = \lambda$ has the following expansion:
\begin{equation}
\begin{aligned} 
\psi(\lambda) &= \psi(\hat\lambda) + \E\left[ -\sum_{u \le t} \hat S_t(X,a) \frac{\hat S_{u-}(X,a)}{\hat S_{u}(X,a)} \Big( \lambda_u(X,a) - \hat \lambda_u(X,a) \Big) \right]
\\&\qquad+\E\left[-\sum_{u \le t} \frac{\hat S_t(X,a)}{\hat S_u(X,a)} \Big( S_{u-}(X,a) - \hat S_{u-}(X,a) \Big) \Big( \lambda_u(X,a) - \hat \lambda_u(X,a) \Big) \right]
\end{aligned}
\label{eq:expansion-full}
\end{equation}

Furthermore, under Assumption \ref{ass:rate} and Lemma \ref{lem:rate_S}, the second term is $o_p(n^{-1/2})$, therefore: 
\begin{align*}
\psi(\hat\lambda) &= \psi(\lambda) 
\\&\quad+ \E\left[ \sum_{u \le t} \hat S_t(X,a) \frac{\hat S_{u-}(X,a)}{\hat S_{u}(X,a)} \Big( \lambda_u(X,a) - \hat \lambda_u(X,a) \Big) \right] 
\\&\quad+ o_p(n^{-1/2})
\end{align*}

\end{lemma}
\begin{remark}
\label{remark:riesz-rep-explicit}
Since the second term in the expansion of $\psi(\lambda)$, shown in \eqref{eq:expansion-full}, is linear in $(\lambda-\hat\lambda)$ and the third is higher order, this lemma shows that the second term is indeed the derivative $d\psi(\hat\lambda)(\lambda - \hat\lambda)$ that appears in \eqref{eq:expansion}.
\end{remark}

\begin{lemma}
\label{lemma:negligible-imbalance}
Suppose the assumptions of Theorem~\ref{thm:efficiency} are satisfied.
\begin{align*}
&\frac{1}{n}\sum_{i=1}^n  \sum_{u \le t}\hat\gamma_{iu} \{ Y^u_i - \hat\lambda_u(Z_i)\} - d\psi(\hat \lambda)(\lambda- \hat\lambda)\\
&= \frac{1}{n}\sum_{i=1}^n  \sum_{u \le t}\Big[\one(A_i=a, G^u_i=1)\hat r_u(X_i,a) \omega_u(X_i,a) \{ Y^u_i - \lambda_u(Z_i) \}\Big] + o_p(n^{-1/2}).
\end{align*}
This remains true if Assumption~\ref{ass:rate}'s $o_p(n^{-1/4})$ rate assumption is weakened to an $o_p(1)$ rate.
\end{lemma}

\begin{lemma}
\label{lemma:eif-continuity}
Suppose our overlap assumption, Assumption~\ref{ass:overlap}, is satisfied. Then the influence function $\phi$ is mean-square continuous as a function of $\lambda$, i.e.,
\begin{align*} 
\phi(\lambda)(O_i)
&=\left\{ S_t(X_i,a) - \E \left[S_t(X,a)\right]\right\} 
\\&\qquad + \sum_{u\le t}\Big[r_u(X_i,a)  \one(A_i=a, G_{iu}=1)\omega_u(X_i,a)\left(\one(\td T_i=u,E_i=1) - \lambda_u(X_i,a)\right) \Big]
\end{align*}
satisfies $\E \{ \phi(\hat\lambda)-\phi(\lambda)\}^2 \to 0$ if $\hat\lambda$ and $\lambda$ are two hazards,
with corresponding survival curves $\hat S$ and $S$ and ratios $\hat r$ and $r$, that converge in the sense that
$\norm{\hat \lambda-\lambda}_{L_2(P)} \to 0$.
\end{lemma}

\subsection{Lemma Proofs}

\begin{proof}[Proof of Lemma~\ref{lem:rate_S}]
For all $t \in \mc T$
\begin{align}
\begin{split}
&\hat S_t(X,a) - S_t(X,a) \\&= \hat S_t(X,a) \left(1 - \frac{S_t(X,a)}{\hat S_t(X,a)}\right)\\
&= \hat S_t(X,a)\sum_{u \le t} \left(\frac{S_{u-}(X,a)}{\hat S_{u-}(X,a)} - \frac{S_u(X,a)}{\hat S_u(X,a)}\right)\\
&= \hat S_t(X,a)\sum_{u \le t}\left[ \frac{S_{u-}(X,a)}{\hat S_{u-}(X,a)}\frac{\hat S_{u-}(X,a)}{\hat S_u(X,a)}\left(\frac{\hat S_u(X,a)}{\hat S_{u-}(X,a)} - \frac{S_u(X,a)}{S_{u-}(X,a)}\right)\right]\\
&= \hat S_t(X,a)\sum_{u \le t} \left[\frac{S_{u-}(X,a)}{\hat S_{u-}(X,a)}\frac{\hat S_{u-}(X,a)}{\hat S_u(X,a)}\left((1 - \hat \lambda_u(X,a)) - (1 - \lambda_u(X,a))\right)\right]\\
&= \sum_{u \le t}\hat S_t(X,a) \frac{S_{u-}(X,a)}{\hat S_{u-}(X,a)}\frac{\hat S_{u-}(X,a)}{\hat S_u(X,a)}\left(\lambda_u(X,a) - \hat \lambda_u(X,a)\right)
\end{split}
\end{align}
therefore
\begin{align*}
&\norm{S_t(X,a) - \hat S_t(X,a)}_{L_2(P)}\\
&=\left\|\sum_{u \le t}\hat S_t(X,a) \frac{S_{u-}(X,a)}{\hat S_{u-}(X,a)}\frac{\hat S_{u-}(X,a)}{\hat S_u(X,a)}
\left(\lambda_u(X,a) - \hat \lambda_u(X,a)\right)\right\|_{L_2(P)}\\
&\le |\mc T| \max_{u \le t}\left\|\hat S_t(X,a) \frac{S_{u-}(X,a)}{\hat S_{u-}(X,a)}\frac{\hat S_{u-}(X,a)}{\hat S_u(X,a)}\left(\lambda_u(X,a) - \hat \lambda_u(X,a)\right)\right\|_{L_2(P)}\\
&\le |\mc T| \max_{u \le t}\norm{\hat S_t(X,a) \frac{S_{u-}(X,a)}{\hat S_u(X,a)}\left(\lambda_u(X,a) - \hat \lambda_u(X,a)\right)}_{L_2(P)}\\&\le |\mc T| \max_{u \le t}\norm{\lambda_u(X,a) - \hat \lambda_u(X,a)}_{L_2(P)}
\end{align*}
since $0 \le \frac{\hat S_t(X,a)S_{u-}(X,a)}{\hat S_u(X,a)} \le \frac{\hat S_t(X,a)}{\hat S_u(X,a)} \le 1$.
\end{proof}

\begin{proof}[Proof of Lemma~\ref{lemma:linearization}]

We expand each term of the decomposition in Lemma \ref{lem:rate_S} around the approximation $S_{u-}(X,a) / \hat S_{u-}(X,a) \approx 1$,
\begin{align*}
&\hat S_t(X,a) \frac{S_{u-}(X,a)}{\hat S_{u-}(X,a)}\frac{\hat S_{u-}(X,a)}{\hat S_u(X,a)}\left(\lambda_u(X,a) - \hat \lambda_u(X,a)\right) \\
&= \hat S_t(X,a) \frac{\hat S_{u-}(X,a)}{\hat S_u(X,a)}\left(\lambda_u(X,a) - \hat\lambda_u(X,a)\right) \\
&\qquad+\hat S_t(X,a) \left( \frac{S_{u-}(X,a)}{\hat S_{u-}(X,a)} - 1 \right) \frac{\hat S_{u-}(X,a)}{\hat S_u(X,a)}\left(\lambda_u(X,a) - \hat\lambda_u(X,a)\right) \\ 
&= \hat S_t(X,a) \frac{\hat S_{u-}(X,a)}{\hat S_u(X,a)}\left(\lambda_u(X,a) - \hat\lambda_u(X,a)\right) \\
&\qquad+ \frac{\hat S_t(X,a)}{\hat S_u(X,a)} \left( S_{u-}(X,a) - \hat S_{u-}(X,a)\right)\left(\lambda_u(X,a) - \hat\lambda_u(X,a)\right)
\end{align*}

therefore
\begin{equation}
\begin{aligned}
&\psi(\hat \lambda) - \psi(\lambda) \\&= \E \left[ \hat S_t(X,a) - S_t(X,a) \right] \\
&= \E\left[\sum_{u \le t}\hat S_t(X,a) \frac{S_{u-}(X,a)}{\hat S_{u-}(X,a)}\frac{\hat S_{u-}(X,a)}{\hat S_u(X,a)}\left(\lambda_u(X,a) - \hat \lambda_u(X,a)\right)\right]\text{~(By Lemma \ref{lem:rate_S})}\\
&= \E\left[\sum_{u \le t} \hat S_t(X,a) \frac{\hat S_{u-}(X,a)}{\hat S_u(X,a)}\left(\lambda_u(X,a) - \hat\lambda_u(X,a)\right)\right]
\\ &\quad + \E\left[\sum_{u \le t} \frac{\hat S_t(X,a)}{\hat S_u(X,a)} \left( S_{u-}(X,a) - \hat S_{u-}(X,a)\right)\left(\lambda_u(X,a) - \hat\lambda_u(X,a)\right)\right]
\end{aligned}
\end{equation}

We now argue that the second term is $o(n^{-1/2})$:
{\allowdisplaybreaks\begin{align*}
&\left|\E\left[\sum_{u \le t}\frac{\hat S_t(X,a)}{\hat S_u(X,a)} \left( S_{u-}(X,a) - \hat S_{u-}(X,a)\right)\left(\lambda_u(X,a) - \hat\lambda_u(X,a)\right)\right]\right|\\
&\le\left\|\sum_{u \le t}\frac{\hat S_t(X,a)}{\hat S_u(X,a)} \left( S_{u-}(X,a) - \hat S_{u-}(X,a)\right)\left(\lambda_u(X,a) - \hat\lambda_u(X,a)\right)\right\|_{L_2(P)}\\
&\le |\mc T|\max_{u\le t}\Big\|\left( S_{u-}(X,a) - \hat S_{u-}(X,a)\right)\left(\lambda_u(X,a) - \hat\lambda_u(X,a)\right)\Big\|_{L_2(P)}\\
&\le |\mc T|\max_{u\le t}\norm{S_{u-}(X,a) - \hat S_{u-}(X,a)}_{L_2(P)}\norm{\lambda_u(X,a) - \hat\lambda_u(X,a)}_{L_2(P)}
\end{align*}
We then use the bound for $S$ in Lemma \ref{lem:rate_S} to get:
\begin{align*}
&\left|\E\left[\sum_{u \le t}\frac{\hat S_t(X,a)}{\hat S_u(X,a)} \left( S_{u-}(X,a) - \hat S_{u-}(X,a)\right)\left(\lambda_u(X,a) - \hat\lambda_u(X,a)\right)\right]\right|\\
&\le |\mc T|^2 \max_{u \le t} \norm{\lambda_u(X,a) - \hat\lambda_u(X,a)}^2_{L_2(P)}\\
&= o_p(n^{-1/2}) \qquad \text{~(By Assumption \ref{ass:rate})}
\end{align*}}
\end{proof}
\begin{proof}[Proof of Lemma~\ref{lemma:negligible-imbalance}]

We will establish this result timestep-by-timestep. That is, we will show that for all $u$,
\begin{equation}
\label{eq:negligible-imbalance-claim}
\begin{aligned}
&\frac{1}{n}\sum_{i=1}^n  \hat \gamma_{iu}\{ Y_{iu} - \hat\lambda_u(X_i,A_i,G_i) \} - \E[\hat r_u(X,a) (\lambda_u -\hat\lambda_u)(X,a,1)] \\
&= \frac{1}{n}\sum_{i=1}^n \gamma_{u}(X_i,A_i,G_i) \{ Y_{iu} - \lambda(X_i,A_i,G_i) \} + o_p(n^{-1/2}).
\end{aligned}
\end{equation}

We will do this by relating it to the error of an augmented minimax linear estimator of $\psi(\lambda_u) = \E[h(X,A,G,\lambda_u)]$ for $h(X,A,G,\eta)=\hat r_u(X,a)\eta(X,a,1)$.
Lemma~\ref{thm:optimization-decomposition} characterizes our $u$-timestep weights $\hat\gamma_u = \hat\gamma_{u1} \ldots \hat \gamma_{un}$ as the weights in that estimator, i.e. the weights described in \citet[Theorem 1]{hirshberg2021augmented} specialized for this instance of $h$.  
\citet[Theorem 2]{hirshberg2021augmented} establishes a bound on the following quantity.
\begin{equation}
\label{eq:if-difference}
\begin{aligned}
&\frac{1}{n}\sum_{i=1}^n h(X_i,A_i,G_i, \hat\lambda_u) + \frac{1}{n}\sum_{i=1}^n  \hat \gamma_{iu}\{ Y_{iu} - \hat\lambda_u(X_i,A_i,G_i) \} - \frac{1}{n}\sum_{i=1}^n \left[ h(X_i,A_i,G_i, \lambda_u) + \tilde\gamma(X_i,A_i,G_i)\{ Y_{iu} - \lambda_u(X_i,A_i,G_i) \} \right] \\
&= \frac{1}{n}\sum_{i=1}^n h(X_i,A_i,G_i, \hat\lambda_u-\lambda_u) + \frac{1}{n}\sum_{i=1}^n  \hat \gamma_{iu} \{ Y_{iu} - \hat \lambda_u(X_i,A_i,G_i) \} - \frac{1}{n}\sum_{i=1}^n \tilde\gamma(X_i,A_i,G_i)\{ Y_{iu} - \lambda_u(X_i,A_i,G_i) \} 
\end{aligned}
\end{equation}
Here $\tilde\gamma$ is a regularized projection of $\gamma_u$. 
\begin{equation}
\label{eq:tilde-gamma}
\tilde\gamma = \argmin_{g}\norm{ g - \gamma_u}_{L_2(P_n)}^2 + \frac{\sigma^2}{n}\norm{g}^2
\end{equation}
 Suppose this quantity \eqref{eq:if-difference} is $o_p(n^{-1/2})$, so the claim \eqref{eq:negligible-imbalance-claim} is equivalent to a version in which we subtract this quantity from the expression on its left side. Here, via arithmetic, is this supposedly equivalent claim.
\begin{equation}
\label{eq:reduced-claim}
\begin{aligned}
&\frac{1}{n}\sum_{i=1}^n \tilde\gamma(X_i,A_i,G_i)\{ Y_{iu} - \lambda_u(X_i,A_i,G_i) \} + \frac{1}{n}\sum_{i=1}^n h(X_i,A_i,G_i, \lambda_u-\hat\lambda_u) - \E[ h(X_i,A_i,G_i, \lambda_u-\hat\lambda_u) ] \\
&=
\frac{1}{n}\sum_{i=1}^n \gamma_{u}(X_i,A_i,G_i) \{ Y_{iu} - \lambda_u(X_i,A_i,G_i) \} + o_p(n^{-1/2}).
\end{aligned}
\end{equation}
Thus, to conclude that the claim holds, it's sufficient to show that three things are $o_p(n^{-1/2})$: \textbf{Term 1}, \eqref{eq:if-difference}; \textbf{Term 2}, the difference of the first terms on each side of \eqref{eq:reduced-claim};
and \textbf{Term 3}, the second term on the left side of \eqref{eq:reduced-claim}. 

\paragraph{Term 3.}
 The second term on the left side of \eqref{eq:reduced-claim} is a sum of iid mean zero terms:
\begin{equation}
\frac1n \sum_{i=1}^n Z_i-\E Z_i = o_p(n^{-1/2}) \quad \text{ for } \quad Z_i = \hat r_u(X_i,a) (\lambda_u -\hat\lambda_u)(X_i,a,1). 
\end{equation}
Thus, its variance is $1/n$ times the variance of an individual term, and to show that it is $o_p(n^{-1/2})$, it suffices to show that the variance of an individual term tends to zero as $n \to \infty$. Given our assumptions, this follows from H\"older's inequality. 
\begin{equation}
\begin{aligned}
\Var[Z_i] \le \E[Z_i^2] 
&= \E[ \hat r_u(X_i,a)^2 (\lambda_u -\hat\lambda_u)(X_i,a,1)^2] \\ 
&\le \norm{r_u(\cdot,a)^2}_{\infty}\norm{(\lambda_u - \hat\lambda_u)(\cdot,a,1)}_{L_1(P)} \\
&= \norm{r_u(\cdot,a)}_\infty^2 \norm{(\lambda_u-\hat\lambda_u)(\cdot,a,1)}_{L_2(P)}.
\end{aligned}
\end{equation}
The first factor here is less than one, as $-\hat r_u$ is the product of a probability, $\hat S_{u-}$, and a probability ratio $\hat S_t /\hat S_u$ that is also less than one because $u \le t$ and our estimated survival probability is decreasing in time. And the second is $o(1)$ by assumption.

\paragraph{Term 2.}
The difference of the first terms on each side of \eqref{eq:reduced-claim} is also
a sum of independent mean zero terms:
\begin{equation}
\frac{1}{n}\sum_{i=1}^n \{\tilde\gamma(X_i,A_i,G_i)-\gamma_u(X_i,A_i,G_i)\}\{ Y_{iu} - \lambda_u(X_i,A_i,G_i) \}.
\end{equation}
In particular, these terms are mean zero \emph{conditional on} $\mathcal{F} = \{ X_i,A_i,G_i : i = 1 \ldots n\}$. 
We will show that its conditional variance is $o_p(1/n)$. This is $1/n$ times the average conditional term variance, so as above it suffices to show this goes to zero. And again, we do this via H\"older's inequality. Because $Y_{iu}$ is an indicator and $\lambda_u$ is a probability, the magnitude of their difference is bounded by $1$. Thus,
\begin{equation}
\begin{aligned}
&\frac{1}{n}\sum_{i=1}^n \E[ \{\tilde\gamma(X_i,A_i,G_i)-\gamma_u(X_i,A_i,G_i)\}^2\{ Y_{iu} - \lambda_u(X_i,A_i,G_i) \}^2 \mid \mathcal{F}_u] \\
&\le \frac{1}{n}\sum_{i=1}^n \E[ \{\tilde\gamma(X_i,A_i,G_i)-\gamma_u(X_i,A_i,G_i)\}^2 \mid \mathcal{F}_u ] \\
&=  \frac{1}{n}\sum_{i=1}^n \{\tilde\gamma(X_i,A_i,G_i)-\gamma_u(X_i,A_i,G_i)\}^2  \quad \text{because this is $\mathcal{F}_u$-measurable} \\
&= \norm{\tilde\gamma-\gamma}_{L_2(P_n)}^2 
\end{aligned}
\end{equation}
This is bounded by the value of the optimization problem \eqref{eq:tilde-gamma} that $\tilde\gamma$ solves, 
so it suffices to show that this value tends to zero. That is, that there exists a sequence $g_n$ converging to $\gamma_u$
in $\norm{\cdot}_{L_2(P_n)}$ with $\norm{g_n} = o(\sqrt{n})$. Suppose, for a moment, that $\gamma_u$ did not vary with $n$.
Then because $\gamma_u$ is an element of what is called the \emph{tangent space} in \citet{hirshberg2021augmented},
i.e. the closure the span of the unit ball of $\norm{\cdot}$ or equivalently the set of limits of sequences $g_1,g_2,\ldots$
with $\norm{g_n} < \infty$ for all $n$, there must be a sequence $g_n$ with $\norm{g_n-\gamma_u}_{L_2(P)} \to 0$
and $\norm{g_n} = o(\sqrt{n})$. We can take any sequence tending toward $\gamma_u$ and choose a subsequence with 
$\norm{g_n}=o(\sqrt{n})$. It follows by the law of the law of large numbers that 
$\norm{g_n}_{L_2(P_n)} \to \norm{g_n}_{L_2(P)}$, so the value of \eqref{eq:tilde-gamma}
does tend to zero. This argument extends naturally to the case that $\gamma_u$ varies with $n$ but 
converges in mean square to some limit $\gamma_u^\star$. By the same argument as before, there exists a sequence $g_n$ 
with $\norm{g_n} = o(\sqrt{n})$  converges in mean-square to $\gamma_u^\star$ and the triangle inequality 
implies that it converges to $\gamma_u$.


\paragraph{Term 1.}
\citet[Theorem 2]{hirshberg2021augmented} establishes a bound on \eqref{eq:if-difference}
that involves several measures of the complexity of the unit ball $\F$ of the norm $\norm{\cdot}$.
We will follow its notation in letting $\Fr$ be the neighborhood $\{ f \in \F : \norm{f}_{L_2(P)} \le r \}$
and $h_{\gamma_u}(\cdot, \Fr)$ and refer to the set $\{h(\cdot, f) - \gamma_u f : f \in \Fr\}$. The first is $R_n(\Fr)$, the Rademacher complexity of the neighborhood $\Fr$. And the second is $h_\gamma(\cdot, \F_r)$. Using the explicit characterization of $\gamma_u$ \eqref{eq:derivative-2} we can write that as follows in terms of a Rademacher sequence $\varepsilon_1\ldots\varepsilon_n$.
\begin{equation}
 R_n(h_\gamma(\cdot, \F_r)) = \E \sup_{f \in \Fr} \frac1n\sum_{i=1}^n \varepsilon_i \hat r(X_i, a)\left\{ 1 - \omega(A_i, G_i) \right\} f(X_i,a,1) \quad \text{for} \quad \omega(A_i,G_i) = \frac{1(A_i=a, G_i^u=1)}{P(A=a, G^u=1)}.
\end{equation}
Because $\hat r(X_i,a)$ is, as discussed above bounded by 1, the contraction theorem for Rademacher processes implies that this is bounded by a version without it, i.e.,
\begin{equation}
\label{eq:contracted-out-r}
 R_n(h_\gamma(\cdot, \F_r)) \le \E \sup_{f \in \Fr} \frac1n\sum_{i=1}^n \varepsilon_i \left\{ 1 - \omega(A_i, G_i) \right\} f(X_i,a,1)
\end{equation}
At this point, we are in a position where the proof of \citet[Theorem 1]{hirshberg2021augmented}
applies without modification. In that proof, $R_n(\F_r)$ and \eqref{eq:contracted-out-r} are bounded using the same assumptions we make and those bounds, in conjunction with assumed bounds on $\hat \lambda-\lambda$,
establish that \citet[Theorem 2]{hirshberg2021augmented} implies an $o_p(n^{-1/2})$ bound on \eqref{eq:if-difference}. In particular, the assumptions used are that $\lVert\hat\lambda_u-\lambda_u\rVert = O_p(1)$
(Assumption~\ref{ass:strongnorm-bounded}); $\lVert\hat\lambda_u - \lambda_u\rVert_{L_2(P_n)} = o_p(1)$, 
follows from our assumption that $\hat\lambda$ is $L_2(P)$ consistent, because it is cross-fit, via the law of large numbers\footnote{Or, without cross-fitting, as a consequence of $\F$ being Donsker}; that $\F$ and $\omega \F$ are Donsker (Assumption~\ref{ass:donskerity}); and that $\omega$ is square-integrable (Assumption~\ref{ass:overlap}).
\end{proof}

\begin{proof}[Proof of Lemma~\ref{lemma:eif-continuity}]

{\allowdisplaybreaks\begin{align*} 
&\phi(\hat\lambda)(O_i) - \phi(\lambda)(O_i) 
\\&=\left( \hat{S}_t(X_i,a) - E\left[\hat S_t(X,a)\right]\right) \\
&\qquad+ \sum_{u\le t}\Big[\hat{r}_u(X_i,a)  \one(A_i=a, G_{iu}=1)\omega_u(X_i,a)\left(\one(\td T_i=u,E_i=1) - \lambda_u(X_i,a)\right)\Big] 
\\& \qquad- \left(S_t(X_i,a) - E\left[ S_t(X,a)\right]\right) 
\\& \qquad - \sum_{u\le t}\Big[r_u(X_i,a)  \one(A_i=a, G_{iu}=1)\omega_u(X_i,a)\left(\one(\td T_i=u,E_i=1) - \lambda_u(X_i,a)\right)\Big] 
\\&= (\hat{S}_t(X_i,a) - S_t(X_i,a)) 
\\&\qquad - \left(E\left[\hat S_t(X,a)\right] - E\left[S_t(X,a)\right]\right) 
\\& \qquad + \sum_{u\le t} \Big[(\hat{r}_u(X_i,a) - r_u(X_i,a))\one(A_i=a, G_{iu} = 1)\omega_u(X_i,a)\left(\one(\td T_i=u,E_i=1) - \lambda_u(X_i,a)\right)\Big]
\end{align*}}
For simplicity, we drop the $(X_i,a)$ when writing functions $\hat{S},S,\hat{r}_u,r_u,\hat{\lambda}_u$, and $\lambda_u$.

We first consider the term $\hat{S}_t - S_t$. From Lemma \ref{lem:rate_S}, we can write
\begin{align*}
\hat S_t - S_t &= \sum_{u \le t}\hat S_t \frac{S_{u-}}{\hat S_{u-}}\frac{\hat S_{u-}}{\hat S_u}\left(\lambda_u - \hat \lambda_u\right) 
\\&= \sum_{u \le t}\hat S_t\frac{ S_{u-}}{\hat S_u}\left(\lambda_u - \hat \lambda_u\right) 
\end{align*}

It is obvious that $0 \le \frac{\hat S_tS_{u-}}{\hat S_u} \le \frac{\hat S_t}{\hat S_u} \le 1$. Applying Holder's inequality and the given condition $\norm{\hat \lambda-\lambda}_{L_2(P)} \to 0$, we can imply that $\hat S \rightarrow S$.

Since $\hat r_u$ and $r_u$ are continuous functions of $\hat S$ and $S$, respectively, and $\hat S\rightarrow S$ as shown above, $\hat r_u \rightarrow r_u$. Then for the last term of the $\phi(\hat\lambda)(O_i) - \phi(\lambda)(O_i)$, we can continue applying Holder's inequality and further conclude that $\E \{ \phi(\hat\lambda)-\phi(\lambda)\}^2 \to 0$.
\end{proof}
\newpage
\subsection{A Sketch of Theorem \ref{thm:efficiency}}
\label{sec:sketch}
Our proof of Theorem~\ref{thm:efficiency} uses results from \citet{hirshberg2021augmented} to do some heavy lifting. For the sake of self-containedness,
we will sketch the main ideas of the argument we'd use to prove it from scratch. We use a more detailed decomposition of the error $\hat\psi - \psi(\lambda)$ as follows:
{\allowdisplaybreaks\begin{equation}
\begin{aligned} 
&\hat \psi - \psi(\lambda)
\\&= \psi_n(\hat\lambda) + \frac{1}{n}\sum_{i=1}^n  \sum_{u\le t} \hat\gamma_{iu} \{ Y_{iu} - \hat\lambda_u(Z_i)\} - \psi(\lambda) 
\\
&= \left\{\psi_n(\hat\lambda) - \psi(\hat\lambda)\right\} 
+ \frac{1}{n}\sum_{i=1}^n \sum_{u \le t}\hat \gamma_{iu}(Y_{iu} - \lambda(Z_i))
\\
&\quad+ \frac{1}{n}\sum_{i=1}^n  \sum_{u\le t} \hat\gamma_{iu} ( \lambda_{iu} - \hat\lambda_u(Z_i) ) 
\\&\quad- \frac{1}{n}\sum_{i=1}^n\sum_{u \le t} \hat r_u(X_i, a)(\lambda(Z_i) - \hat \lambda(Z_i))\\
&\quad+ \frac{1}{n}\sum_{i=1}^n\sum_{u\le t}\hat r_u(X_i, a)(\lambda(Z_i) - \hat\lambda(Z_i))
\\&\quad- d\psi(\hat\lambda)(\lambda - \hat\lambda) \\
&\quad + \left\{ \psi(\hat\lambda) + d\psi(\hat\lambda)(\lambda - \hat\lambda) - \psi(\lambda) \right\}.
\end{aligned}
\end{equation}}
We sketch the analysis of each of the 4 terms above:
\begin{enumerate}
\item The first term converges to the influence function of the estimator because $\hat \lambda$ and $\hat \gamma$ are convergent. That the latter converges to the population Riesz representer is a consequence of the analysis of the imbalance in the 2nd term below.
\item The 2nd term is the imbalance term motivated by the approximation \ref{eq:sample-balance}, and our optimization problem \ref{eq:weight-optimization} directly controls it. This is exactly where we borrow the covariate-balancing analysis of \cite{hirshberg2021augmented} to our problem, noting that they have similar structure.
\item The 3rd term is the difference of the sample-average derivative and its expectation, can be shown to be $o(n^{-1/2})$ because each term of the mean $(\sum_{u\le t}\hat r_u(X_i, a)(\lambda(Z_i) - \hat \lambda(Z_i)) - d\psi(\hat\lambda)(\lambda - \hat \lambda))$ has mean 0 and variance $o(1)$ as consequence of the convergence of $\hat \lambda$.
\item The 4th term is the 2nd-order remainder as before and is $o(n^{-1/2})$.
\end{enumerate}

Overall, we see again that $\hat\psi - \psi(\lambda) = \sum_{i=1}^n\phi(O_i) + o(n^{1/2})$.

\end{document}